\documentclass[a4paper,onecolumn,11pt,accepted=2021-06-25]{quantumarticle}
\pdfoutput=1
\usepackage[utf8]{inputenc}
\usepackage[english]{babel}
\usepackage[T1]{fontenc}

\usepackage[usenames,dvipsnames]{color}
\usepackage[hyperfootnotes=false]{hyperref}
\hypersetup{
	colorlinks,
	citecolor=quantumviolet,
	linkcolor=quantumviolet,
	urlcolor=quantumviolet}
\usepackage{tikz}
\usepackage{lipsum}
\usepackage[numbers,sort&compress]{natbib}

\usepackage{mathtools}
\usepackage{graphicx}
\usepackage{amssymb}
\usepackage{amsmath}
\usepackage{amsthm}
\usepackage{amsfonts}
\usepackage{color}
\usepackage{xcolor}
\usepackage{braket}
\usepackage{soul}

\usepackage{subfigure}
\usepackage{booktabs}
\usepackage{hyperref}

\newtheorem{lemma}{Lemma}
\newtheorem{example}{Example}
\usepackage{dsfont}
\usepackage{algorithm, algorithmic}

\DeclareMathOperator{\tr}{\mathrm{tr}}

\newcommand{\Mod}[1]{\ (\mathrm{mod}\ #1)}

\newcommand{\HS}{\mathcal{H}}
\newcommand{\mcb}[1]{\mathcal{B}_{#1}}
\newcommand{\DWC}{\mathcal{N}_{\mathrm{dwc}}}
\newcommand{\DWCs}[1]{\DWC \left( #1\right)}
\newcommand{\Qset}[1]{\mathbb{Q}_{#1}}
\newcommand{\V}[1]{\mathbf{#1}}
\newcommand{\M}[1]{\mathbf{#1}}
\newcommand{\dbar}[1]{\bar{\bar{#1}}}
\newcommand{\mgnt}[3]{f\left(\bar{\bar{#1}}; \, #2, #3 \right)}
\newcommand{\sumDWCS}[3]{\sum_{{\dbar{#1}:\atop  \mgnt{#1}{#2}{#3}= \ell }}} 
\newcommand{\sNorm}[1]{\left|{#1}\right|}
\DeclareMathOperator{\E}{\mathds{E}}
\newcommand{\EX}[1]{\E\left\{{#1}\right\}}
\newcommand{\we}[1]{\M{W}_{#1}}
\newcommand{\vNorm}[1]{\left\| #1\right\|}

\DeclareMathOperator{\var}{\mathds{V}\mathsf{ar}}
\newcommand{\Var}[1]{\var\left\{{#1}\right\}}

\definecolor{CCTLABgreen}{RGB}{0,128,0}

\begin{document}

\title{Entanglement-Free Parameter Estimation of Generalized Pauli Channels}

\author{Junaid ur Rehman}
\orcid{0000-0002-2933-8609}
\author{Hyundong Shin}
\orcid{0000-0003-3364-8084}
\affiliation{Department of Electronics and Information Convergence Engineering, Kyung Hee University, Korea}
\email{hshin@khu.ac.kr}
\thanks{(corresponding author)}
\maketitle

\begin{abstract}
We propose a parameter estimation protocol for generalized Pauli channels acting on $d$-dimensional Hilbert space. The salient features of the proposed method include product probe states and measurements, the number of measurement configurations linear in $d$, minimal post-processing, and the scaling of the mean square error comparable to that of the entanglement-based parameter estimation scheme for generalized Pauli channels. We also show that while measuring generalized Pauli operators the errors caused by the Pauli noise can be modeled as measurement errors. This makes it possible to utilize the measurement error mitigation framework to mitigate the errors caused by the generalized Pauli channels. We use this result to mitigate noise on the probe states and recover the scaling of the noiseless probes, except with a noise strength-dependent constant factor.
This method of modeling Pauli channel as measurement noise can also be of independent interest in other NISQ tasks, e.g., state tomography problems, variational quantum algorithms, and other channel estimation problems where Pauli measurements have the central role.
\end{abstract}

\section{Introduction}
Second quantum revolution has introduced a wide range of new quantum technologies. Quantum states and channels hold a central role in the efficient and successful implementation of all of these technologies. It is desirable to design our systems-of-interest as close to ideal behaviour as possible. However, environmental effects and nonidealities in designed components inevitably and irrecoverably introduce noise in these systems. A general method to model this noise in system components is through quantum channels. Ideally, one would aim for system components to be noiseless and error free, i.e., involved channels are identity channels.  However, it is almost impossible to design noiseless system components. The next best possible scenario is to have a complete knowledge of noise present in the system. That is, to know all the ways in which noise can corrupt the system and lead it to deviate from the intended behaviour. Having a complete knowledge of noise present in system components allows one to efficiently minimize the errors introduced by the noise \cite{CV:20:LNCS, FSW:08:IEEE_IT, XLM:12:CTW}.

Quantum process tomography is the method to identify an unknown quantum dynamical process \cite{CN:97:JMO, PCZ:97:PRL, OPG:04:PRL}. The general method of process tomography is to prepare probe states in different initial states, let them evolve through the quantum process of interest, and then measure the output states with different measurement settings \cite{DP:01:PRL, ML:07:PRA}. A measurement configuration is the specific setting of initial state of probes and measurement settings, i.e., changing the initial state of probe or the measurement setting gives a new measurement configuration. In general, the quantum process tomography is a resource-intensive and experimentally demanding process; standard quantum process tomography of a general quantum channel on $d$-dimensional Hilbert space requires $d^4$ measurement configurations. This stringent requirement of a large number of measurement configurations can be relaxed either by operating on a larger Hilbert space (entangled probes schemes) or by making reasonable assumptions on the channel structure based on the prior knowledge \cite{CDS:05:PRA, CS:16:IJQI, MRL:08:PRA}. Examples of the latter strategy include assumption of rank deficiency \cite{RVB:14:PRB} or modeling the unknown given channel as a parametric class of channels and then estimating the unknown parameters \cite{FCG:11:JPAMT, JWD:08:IEEE_IT, Fuj:04:PRA}. 

Examples of such parametric classes of channels include Pauli qubit channels and their higher-dimensional generalizations including discrete Weyl channels (DWCs) \cite{OP:09:AMH, SC:18:JMP}. Study of Pauli channels and their generalizations is well motivated by several important properties of this class. For example, it is known that \emph{every} unital qubit channel is similar to Pauli qubit channel \cite{PRZ:19:PLA}. Furthermore, several physically important classes of quantum channels are special cases of Pauli channels. Examples include depolarizing, dephasing, bit-flip, and two-Pauli channels. Furthermore, any noise model on a multiqubit system can be modeled as having the form of a Pauli channel \cite{Kni:05:Nat,  HFW:20:NP}. In recent times, some practical methods have been introduced that effectively approximate any noise model as the Pauli channel \cite{ESM:07:Sci, GZ:13:PRA, WE:16:PRA, HFW:20:NP, CXB:20:nQI} {e.g., by twirling via Pauli operators}. Unfortunately, some of the above motivations no longer remain true for the higher dimensional generalizations of Pauli channels \cite{DR:05:JPAMG}. Regardless, generalizations of Pauli channels remain an important and interesting topic of study in the theory of quantum information processing.

Due to their practical relevance and versatility, several researchers have studied the general and specific variants of Pauli channels to devise different strategies for estimating their parameters \cite{FI:03:JPAMG, Hay:10:JPACS, Hay:11:CMP, FCG:11:JPAMT, BHP:12:T_AC, RVH:12:JPAMT, VRH:13:AX, Col:13:PRA, CS:15:PRA, FW:20:ACMTQC,HYF:21:PRXQ, HFW:20:NP}. Of particular interest to us is the entanglement-assisted optimal parameter estimation (OPE) protocol presented in \cite{FI:03:JPAMG}, which is optimal in the sense of Cram\'{e}r-Rao bound, provides the best scaling of mean square error (MSE) in the number of channel uses, requires only a \emph{single} measurement configuration, and deals with the most general case of the generalized Pauli channels without any further assumptions. Experimental realization of this protocol for qubit Pauli channels was given in \cite{CRV:11:PRL}. However, experimental realization of this (and other entanglement-assisted) protocol becomes extremely challenging in the higher-dimensional cases due to difficulties involved in generating, maintaining, and processing higher-dimensional entangled states \cite{LRL:08:PRL, SLS:19:NJP}. 

In this paper, we present a protocol for the parameter estimation of DWCs, which can also be applied on the other generalizations of Pauli channels. The proposed protocol, called the direct parameter estimation of Pauli channels (DPEPC), is solely based on separable states but provides the same scaling of MSE as a function of channel uses as that of the OPE but with a multiplicative factor. Unfortunately, DPEPC requires more than a single measurement configurations. However, extensive numerical examples suggest that the required number of measurement configurations scales linearly with the dimension of the Hilbert space. Additionally, we show that in a system with Pauli measurements, errors caused by a Pauli channel can be efficiently modeled as measurement errors.  Then, the framework of measurement error mitigation can successfully mitigates these errors. We provide numerical examples of this error mitigation by introducing additional depolarizing noise on the probe states and then mitigating its effects by the aforementioned technique. This procedure recovers the original scaling of both DPEPC and OPE except with another noise strength-dependent multiplicative factor, if the noise strength is known. 


The remainder of this paper is organized as follows. In Section~\ref{sec:notations} we set the notations and preliminaries. Section~\ref{sec:DPT} and \ref{sec:NE} provide the protocol and numerical examples of DPEPC for the DWCs, respectively. In Section~\ref{sec:Conc}, we provide the conclusions and future outlook.

\section{Notations and Preliminaries} \label{sec:notations}
A DWC is a qudit generalization of qubit Pauli channels. The DWC acts on a quantum state $\rho$ as
\begin{equation}
	\DWCs{\rho} = \sum_{n = 0}^{d-1} \sum_{m = 0}^{d-1} p_{n,m} \we{n,m} \rho \we{n,m}^{\dagger},
	\label{eq:DWCdef}
\end{equation}
where 
\begin{align}
\we{n,m} = \sum_{k = 0}^{d-1} \omega^{kn} \ket{k} \bra{k + m\Mod{d}}, \quad 0 \leq n,m \leq d-1
\end{align} 
with $\omega = \exp\left(2\pi \dot{\iota}/d\right)$ are $d^2$ discrete Weyl operators on the $d$-dimensional Hilbert space $\HS_d$; $\left\{p_{n,m}\right\}$ form a probability vector and are called the channel parameters. Estimation of $\left\{p_{n,m}\right\}$ of an unknown given DWC is the main objective of the OPE and DPEPC. We denote the set of all Weyl operators on $\HS_d$ by $\mathcal{W}_d$.

For simplicity, we will also utilize a single index notation for discrete Weyl operators and the elements of probability vector of \eqref{eq:DWCdef}, where $\M{V}_{\dbar{k}} = \we{n,m}$ and $q_{\dbar{k}} = p_{n,m}$, with $\dbar{k} = n + md$. There exists an index-based relation between a Weyl operator $\we{a,b}$ and the eigenvectors of another Weyl operator $\we{n,m}$. The relationship was first presented by the authors in \cite{RJK:18:SR} and is formally given in Lemma~\ref{Lemma:1} of the current manuscript. Due to repetitive appearance of index relation $ma - nb \mod{d}$, we define it as $f\left( \dbar{k}; n,m\right)$ where it is understood that $\dbar{k}$ will first be decompressed to the double index notation to calculate $ma - nb \mod{d}$. In particular, $f\left( \dbar{k}; n,m\right) = 0$ if and only if $\we{a,b}$ and $\we{n,m}$ commute. We denote the orthonormal eigenbasis of $\we{n,m}$ by $\mcb{n,m}$. We also define $\Qset{d} = \left\{ 0,1, \cdots , d-1\right\}$.

\section{Direct Parameter Estimation of Pauli Channels}\label{sec:DPT}
{In this section, we outline our protocol for the parameter estimation of Pauli channels. The key idea is the equivalence of DWCs with classical symmetric channels under certain conditions \cite{RJK:18:SR}. By estimating the transition probabilities of emulated classical symmetric channels, we are able to reconstruct the full parameter set of the underlying DWCs. We also explore the quantum error mitigation for mitigating errors caused by noise in the probe states. }

\subsection{Proposed Protocol}
A DWC acts as a classical symmetric channel when the inputs to the channel are the elements of $\mcb{n,m}$, and the measurement at the output is a projective measurement in $\mcb{n,m}$. Then, the transition probabilities of the effective classical channels are given by the following lemma.
\begin{lemma}[\cite{RJK:18:SR}]\label{Lemma:1}
Let $\we{n,m}$ have $d$ distinct eigenvalues and its eigenstate $\ket{i_{n,m}}$ be input to a DWC. Then, the output state is diagonal in $\mcb{n,m}$ and its eigenvalues $\lambda_{\ell}^{n,m}$, $\ell \in \Qset{d}$ are given by
\begin{equation}
\lambda_{\ell}^{n,m} = \sumDWCS{k}{n}{m} q_{\dbar{k}}
\label{eq:lemma_closed}
\end{equation}
where $q_{\dbar{k}}$ are the parameters of the DWC.
\label{lemma:closed}
\end{lemma}
In the context of the simulated classical channel, $\lambda_{\ell}^{n,m}$ is the probability of observing the output state $\ket{\left(i + \ell\right)_{n,m}}$ when the input state to the channel was $\ket{i_{n,m}}$. Due to the orthogonality of the elements of $\mcb{n,m}$ it is possible to obtain a direct estimate on $\lambda_{\ell}^{n,m}$, $\ell \in \Qset{d}$ by utilizing Lemma~\ref{lemma:closed}. Additionally, due to the independence of $\lambda_{\ell}^{n,m}$ from the index $i$ of the input state, the estimates on $\lambda_{\ell}^{n,m}$ for all $\ell$ are obtained simultaneously. That is, for any chosen $\ket{i_{n,m}}$ from $\mcb{n,m}$ and for any $\ell$, $\lambda_{\ell}^{n,m}$ is simply the fraction of times $\ket{\left(i + \ell\right)_{n,m}}$ is measured at the channel output. Therefore, one experiment configuration (fixed input and projective measurement in $\mcb{n,m}$) is sufficient to estimate the complete set of $d$ transition probabilities $\lambda_{\ell}^{n,m}$ for a fixed $\we{n, m}$.

\begin{figure}[t]
	\centering
	\includegraphics[width = 0.6\textwidth]{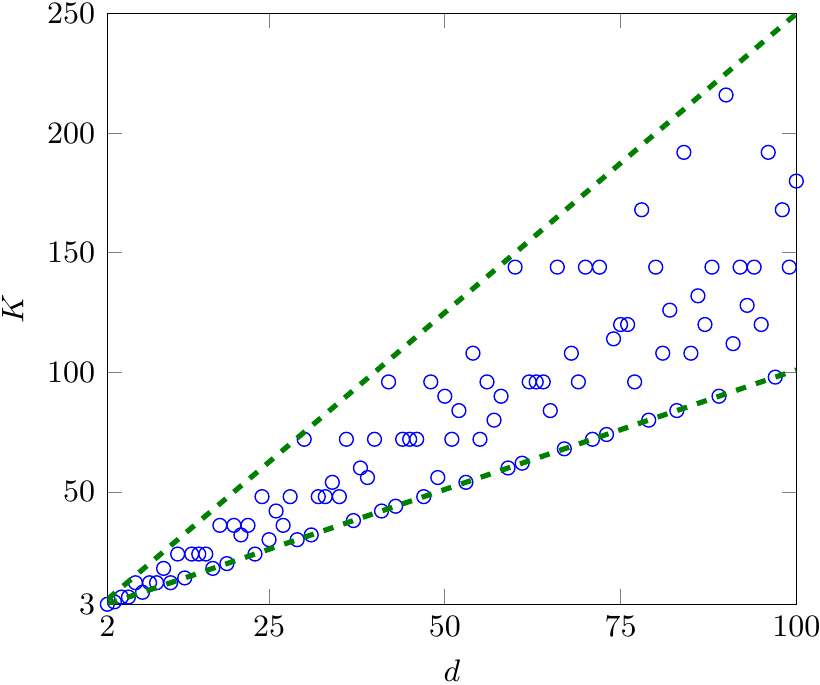}
	\caption{\textbf{Sufficient Measurements for DPEPC.} Number of measurement configurations $K$ vs the dimension $d$ of the Hilbert space for the DPEPC of DWCs. Lower and upper dashed green lines are at $d + 1$ and $d\times 2.5$, respectively.}
	\label{fig:Alg}
\end{figure}

For a fixed $\mcb{n,m}$, \eqref{eq:lemma_closed} provides a set of $d$ simultaneous equations which can be written in the matrix form $\M{A}^{n,m}\V{x} = \V{b}^{n,m}$, where $\V{b}^{n,m}$ (resp. $\V{x}$) is the $d \times 1$ (resp. $d^2 \times 1$) vector with $\lambda_{\ell}^{n,m}$ (resp. $p_{a,b}$ ) as its elements and $\M{A}^{n,m}$ is a $d \times d^2$ matrix with entries defined as
\begin{equation}
	A_{j,\dbar{k}} = \mathbb{I}_{j}\left( \mgnt{k}{n}{m}\right),
\end{equation}
where $\mathbb{I}_{j}\left( i\right)$ is the indicator function defined as
\begin{equation}
	\mathbb{I}_{j}\left( i\right) = 
	\begin{cases}
		1, \text{ if } i = j,\\
		0, \text{ otherwise.}
	\end{cases}
\end{equation}

Once we obtain the estimates on the elements of $\V{b}^{n,m}$, we can attempt to solve the set of equations $\M{A}^{n,m}\V{x} = \V{b}^{n,m}$ to obtain the channel parameters contained in $\V{x}$. However, in order to solve $\M{A}^{n,m}\V{x} = \V{b}^{n,m}$ for a unique $\V{x}$, we need the rank of $\M{A}^{n,m}$ to be $d^2$, which is impossible for our $d\times d^2$ matrix $\M{A}^{n,m}$. Since the summation \eqref{eq:lemma_closed} partitions the elements $q_{\dbar{k}}$ in $d$ disjoint sets of $d$ elements each such that the elements in each set contribute to a particular $\lambda_{\ell}^{n,m}$, the rows of $\M{A}^{n,m}$ are linearly independent. Thus, $\M{A}^{n,m}$ has rank $d$ for any $\we{n,m}$ that has $d$ distinct eigenvalues.

\begin{figure}[t]
	\centering
	\subfigure[~]{
		\includegraphics[width=0.7\textwidth]{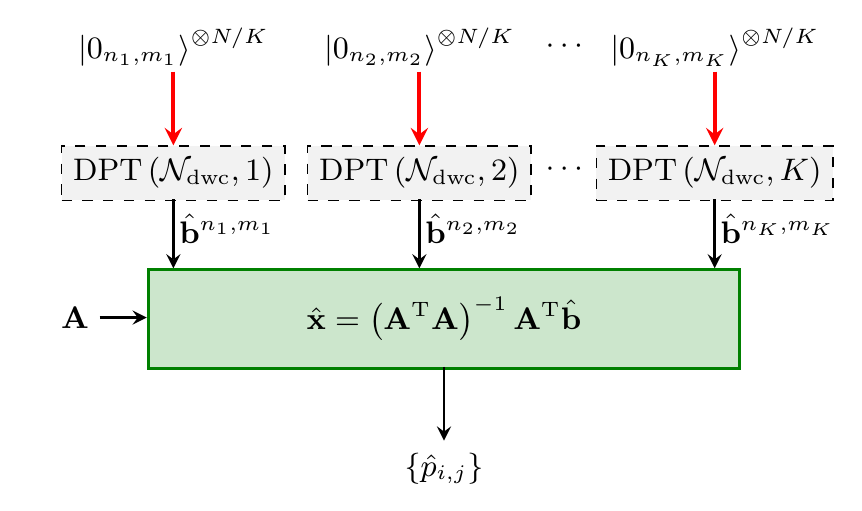}
	}
	~
	\subfigure[~]{
		\includegraphics[width=0.8\textwidth]{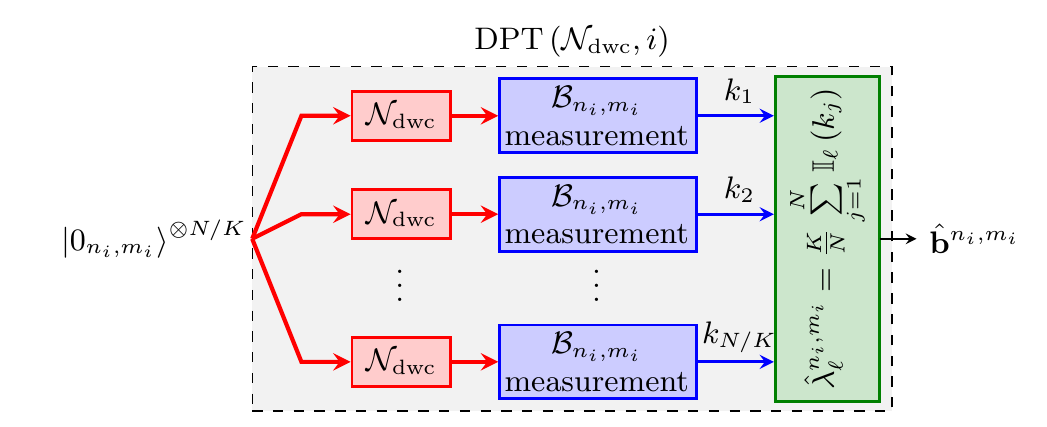}
	}
	\caption{\textbf{DPEPC protocol for DWCs.} (a) For each $i\in \Qset{K+1}\setminus \left\{0\right\}$ $N/K$ copies of each probe state $\ket{0_{n_i, m_i}}$ are input to the block $\mathrm{DPT}\left( \mathcal{N}_{\mathrm{dwc}}, i\right)$, which output estimates $\hat{\V{b}}^{n_i, m_i}$ of $\V{b}^{n_i, m_i}$. Estimates $\hat{p}_{i,j}$ on the channel parameters $p_{i,j}$ are obtained via the method of least squares by utilizing the estimates $\V{b}^{n_i, m_i}$. (b) Structure of the block $\mathrm{DPT}\left( \mathcal{N}_{\mathrm{dwc}}, i\right)$.}
	\label{fig:DPT}
\end{figure}

We can solve this problem of having smaller number of available simultaneous equations than the unknowns in the system by obtaining more equations for different $n,m$ values. That is, we invoke Lemma~\ref{lemma:closed} for $K$ different values of $n$ and $m$ to obtain at least $Kd$ equation in the matrix form
\begin{equation}
	\begin{bmatrix}
	\M{A}^{n_1,m_1} \\ \M{A}^{n_2,m_2} \\ \vdots \\ \M{A}^{n_K,m_K}
	\end{bmatrix}
	\V{x} = 
	\begin{bmatrix}
	\V{b}^{n_1,m_1} \\ \V{b}^{n_2,m_2} \\ \vdots \\ \V{b}^{n_K,m_K}
	\end{bmatrix}.
	\label{eq:A_big}
\end{equation}
We denote the matrix on the left hand side of \eqref{eq:A_big} by $\M{A}^d_{K}$, where the superscript denotes the dimension of the Hilbert space on which the channel operates and the subscript denotes the total number of \emph{non-commuting} $\we{n,m}$ using which Lemma~\ref{lemma:closed} was invoked.\footnote{We require different $\we{n,m}$ to be non-commuting because commuting $\we{n,m}$'s would result into same rows for $\M{A}^{n,m}$ but with different ordering.} The set of corresponding indices of Weyl operators utilized in generating $\M{A}^d_{K}$ is denoted by $\mathcal{W}_{\mathrm{idx}}$.

One would then hope that the system \eqref{eq:A_big} with $K = d$, would have a unique solution. However, we show that the matrix $\M{A}^{d}_{d}$ is still rank deficient for any $d$. First, note that all the elements of the row obtained by summing all rows of any $\M{A}^{n_k,m_k}$ will be 1. Then, one can obtain any row of any  $\M{A}^{n_{k'},m_{k'}}$ by simply subtracting all other rows of $\M{A}^{n_{k'},m_{k'}}$ from the row containing all 1's. Therefore, $\M{A}^{d}_{d}$ despite being of dimension $d^2\times d^2$ is still rank deficient. Therefore, the minimum $K$ such that $\M{A}^{d}_{K}$ has rank $d^2$ for any $d$ is at least $d+1$. In the following, we call an $\M{A}^{d}_{K}$ sufficient if it has rank $d^2$.

{Analytically }obtaining the exact value of the smallest $K$ for an arbitrary $d$ such that $\M{A}_K^d$ is sufficient is difficult. To overcome this difficulty, we algorithmically obtain $\mathcal{W}_{\mathrm{idx}}$.\footnote{The source code is available at \url{https://github.com/junaid572/DPEPC}.} 
{Verbal description of our algorithm is as follows. We first utilize the results from \cite{RJS:19:PRA} to calculate the total number of distinct eigenvalues of all discrete Weyl operators on $\HS_d$. We make a set $\mathcal{W}_d$ of all Weyl operators that have $d$ distinct eigenvalues. Then, we utilize the identity $\we{n, m} \we{p, q} = \omega^{nq - mp}\we{p, q}\we{n, m}$ to identify the commutation relations of operators within $\mathcal{W}_d$. We make subsets of $\mathcal{W}_d$ such that operators within each subset mutually commute. Finally, we obtain $\mathcal{W}_{\mathrm{idx}}$ by choosing one operator each from the commuting subsets of $\mathcal{W}_d$. We verify that $\mathcal{W}_{\mathrm{idx}}$ generates a sufficient $\M{A}^{d}_{K}$ by constructing the corresponding $\M{A}^{d}_{K}$  and verifying that it has rank $d^2$. } 

We used this algorithm for $d$ upto 100, which provides the following insights. For any $d$, an $\M{A}^{d}_{d^2 - 1}$ is always sufficient. That is, for a $d$-dimensional DWC, $d^2 - 1$ measurement configurations are always sufficient to perform the full process tomography. This number can be considerably reduced by utilizing the commutation relations of discrete Weyl operators. {We were able to obtain a sufficient $\M{A}^{d}_{K}$ for $K< d\times 2.5$ for any $\HS_d$ as large as $d = 100$}. Figure~\ref{fig:Alg} shows the required number of measurement configurations $K$ obtained via this algorithm. Furthermore, for any prime $d$, $\M{A}^{d}_{d+1}$ is sufficient. This latter observation is expected to hold beyond the values of $d$ which we numerically checked, since it is not possible to construct a set of more than $d+1$ noncommuting Weyl operators for a prime $d$ \cite{RJS:19:PRA}. Therefore, if $\M{A}^{d}_{d^2 - 1}$ being sufficient for any $d$ is always true, then the sufficiency of $\M{A}^{d}_{d + 1}$ for any prime $d$ also holds everywhere.

Obtaining a sufficient $\M{A}_K^d$ entails constructing the binary matrix $\M{A}_K^d$ as well as identifying the indices $n_i, m_i$, for $1\leq k \leq K$ of Weyl operators whose eigenstates will be utilized for the DPEPC of DWC. Once a sufficient $\M{A}_K^d$ is found for a $d$, the DPEPC of a DWC for $N$ channels uses can be performed as follows. Prepare $\left\lfloor N/K \right\rfloor$ copies of an eigenstate $\ket{s_{n_k, m_k}}$ of $\we{n_k, m_k}$ for every $1\leq k \leq K$ and send them through the channel $\mathcal{N}_{\mathrm{dwc}}$. For every $\ket{s_{n_k, m_k}}$ at input, measure the channel output in $\mcb{n_k,m_k}$ and record the measurement. Measurement outcomes provide an estimate $\hat{\lambda}_{\ell}^{n_k, m_k}$ for all $\lambda_{\ell}^{n_k, m_k}$. Construct the vector $\hat{\V{b}}_{K}^{d}$, which is an estimate on the vector on the right hand side of \eqref{eq:A_big}. Finally, obtain the estimates $\hat{p}_{i,j}$ on channel parameters $p_{i,j}$ by the method of least squares, i.e., 
\begin{equation}
	\hat{\V{x}} = 
		\left( \left(\M{A}_K^d\right)^{\mathrm{T}} \M{A}_K^d \right)^{-1} 
		\left(\M{A}_K^d\right)^{\mathrm{T}} 
		\hat{\V{b}}_{K}^{d},
		\label{eq:LSE}
\end{equation}
where $\left(\cdot\right)^{\mathrm{T}}$ and $\left(\cdot\right)^{-1}$ are the matrix transpose and the matrix inverse operations, respectively. {Note that the inverse in \eqref{eq:LSE} is only dependent on $d$ and the utilized measurement configurations, not on data. Thus,  after fixing $d$ and $K$, we can precompute 
\begin{align}
\M{B}^d_K = \left( \left(\M{A}_K^d\right)^{\mathrm{T}} \M{A}_K^d \right)^{-1} 
\left(\M{A}_K^d\right)^{\mathrm{T}}.
\label{eq:B_matrix}
\end{align}
 All subsequent runs of DPEPC can be completed simply by computing $\hat{\V{x}} = \M{B}^d_K \hat{\V{b}}_{K}^{d}$. In this sense, no matrix inversion is needed in DPEPC.}  Figure~\ref{fig:DPT} depicts the complete DPEPC protocol for DWCs.

It was shown in \cite{RJS:19:PRA} that variance in the estimates on the transition probabilities of Lemma~\ref{lemma:closed} scale with $1/N$, which is same as the scaling of OPE, except with a constant multiplicative factor $K$. We obtain the estimates $\hat{p}_{n,m}$ on the channel parameters $p_{n,m}$ by multiplying the estimates on transition probabilities with a matrix which is independent of $N$. Therefore, we obtain the same scaling in the estimates of $\hat{p}_{n,m}$, i.e., $K/N$.

\begin{figure}[t]
	\centering
	\includegraphics[width = 0.7\textwidth]{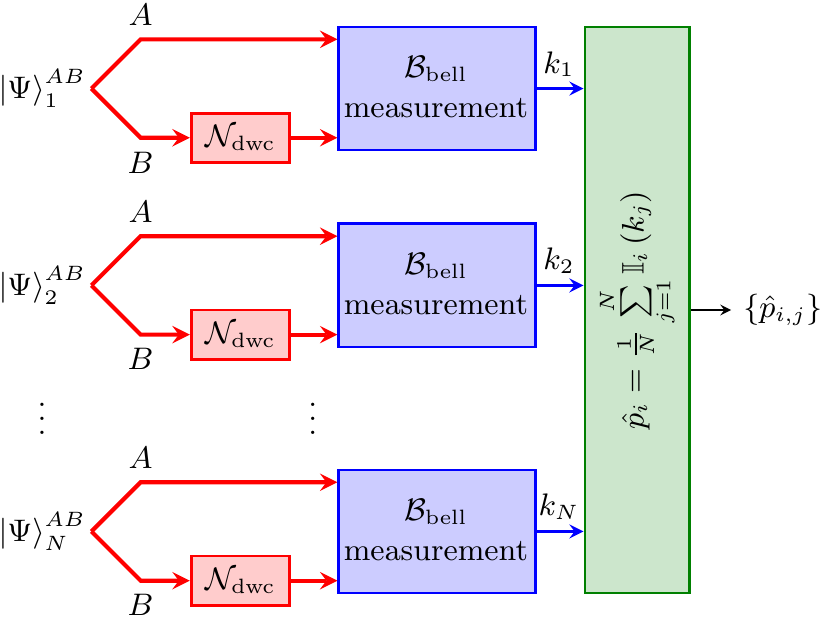}
	\caption{\textbf{OPE protocol of \cite{FI:03:JPAMG}.} $N$ copies of probe state $\ket{\Psi}^{AB}$ are prepared, where $\ket{\Psi}^{AB}\in \HS_d \otimes \HS_d$ is the two-qudit maximally entangled state. One of the qudits is allowed to evolve under $\mathcal{N}_{\mathrm{dwc}}$ and subsequently a joint $\mathcal{B}_{\mathrm{bell}}$ measurement on both qudits in perfomed. Finally, the vector of probabilities is $\hat{p}_{i,j}$ is estimated from the measurement statistics. }
	\label{fig:OPE}
\end{figure}

Before moving to the numerical examples and comparison section, we provide an expository example of DPEPC for $d=2$ DWC. This example not only serves the purpose of exposition but also highlights the salient features of the DPEPC for DWCs.

\begin{example}[DPEPC for the qubit DWC]
	We have 
	\begin{equation}
		\M{A}^{2}_{3} 
		= 
		\begin{bmatrix}
			\M{A}^{0,1}\\
			\M{A}^{1,0}\\
			\M{A}^{1,1}
		\end{bmatrix}
		=
		\begin{bmatrix}
			1 & 1 & 0 & 0\\
			0 & 0 & 1 & 1\\
			1 & 0 & 1 & 0\\
			0 & 1 & 0 & 1\\
			1 & 0 & 0 & 1\\
			0 & 1 & 1 & 0
		\end{bmatrix},
	\end{equation}
$\V{x} = \left[ p_{0,0} \ p_{0,1} \  p_{1,0} \ p_{1,1}\right]^{\mathrm{T}}$, and $\V{b}^{2}_{3} = \left[ \lambda_0^{0,1} \ \lambda_1^{0,1} \ \lambda_0^{1,0} \ \lambda_1^{1,0} \ \lambda_0^{1,1} \ \lambda_1^{1,1}\right]^{\mathrm{T}}$. Three probe states are 
\begin{equation}
\begin{aligned}
	\ket{0_{0,1}} &= \ket{+} = 1/\sqrt{2}\left(\ket{0} + \ket{1}\right) \\ 
	\ket{0_{1,0}} &= \ket{0}, \text{ and }\\
	\ket{0_{1,1}} &= \ket{+\dot{\iota}} = 1/\sqrt{2}\left(\ket{0} + \dot{\iota}\ket{1}\right).
\end{aligned}
\end{equation}
The corresponding measurement settings are the projective measurements in $\mcb{n_i, m_i}$. The estimate $\hat{\lambda}_{\ell}^{n_i, m_i}$ is the relative frequency of outcome $\ket{\ell_{n_i, m_i}}$ when $\ket{0_{n_i, m_i}}$ was input to the channel. Finally, an estimate on the channel parameters is obtained via \eqref{eq:LSE}.
\label{Example:qubitDPT}
\end{example}

\begin{figure}[t!]
	\centering
	\includegraphics[width = 0.7\textwidth]{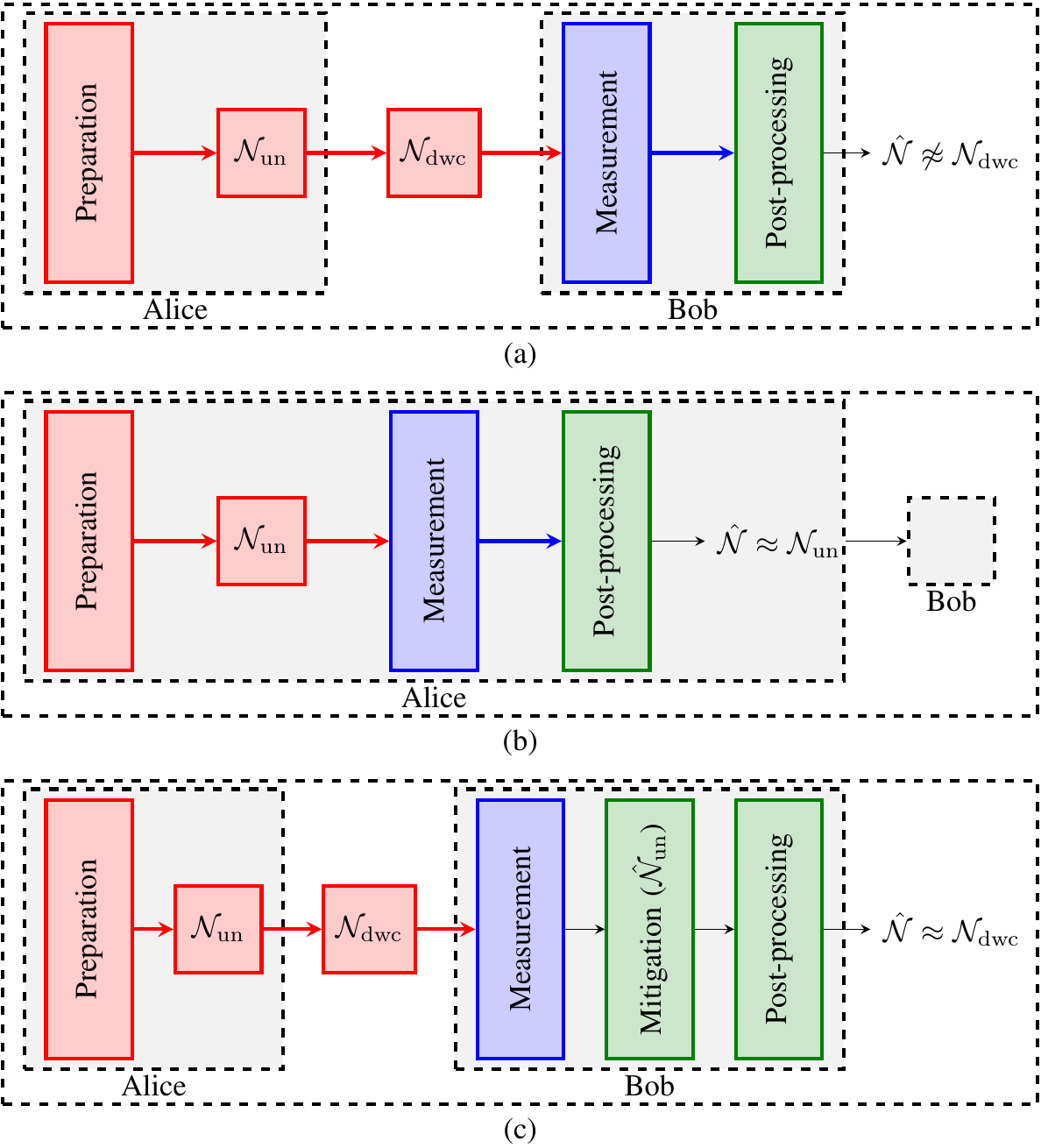}
	\caption{\textbf{DPEPC with Noise}. A practical scenario with unintended noise $\mathcal{N}_{\mathrm{un}}$ on the probe states is depicted in (a). Errors caused by this noise can be mitigated if Alice locally estimates this noise (b), and then Bob applies quantum error mitigation after measurement in a regular run (c). We utilize the measurement error mitigation in this configuration. }
	\label{fig:Pract_Scene}
\end{figure}

We stress that once a sufficient $\M{A}^d_K$ is constructed for a given $d$, that $\M{A}^d_K$ can be utilized for all the subsequent DPEPC experiments for all DWCs operating on $\HS_d$. This also fixes the measurement configurations and the pseudo-inverse of $\M{A}^d_K$ appearing on the right side of \eqref{eq:LSE} for this $d$. Therefore, the DPEPC for DWCs does not involve experiment design, matrix inversion, or optimization of any kind. The DPEPC protocol for DWCs is then to simply perform measurements in $K$ pre-defined measurement configurations and plug-in the frequencies of measurement results in $\V{b}^d_K$ to directly obtain the channel parameters $\V{x}$.

\subsection{Quantum Error Mitigation for DPEPC}
The proposed protocol in the previous subsection relies on the ability to sufficiently isolate the prepared probe states such that the only noisy evolution they go through is the noisy channel under study. However, this isolation might not be possible in practice. An unintended noisy evolution might occur anywhere from preparation to the final measurement. Such a scenario is shown in Figure~\ref{fig:Pract_Scene} where an unintended noise $\mathcal{N}_{\mathrm{un}}$ may corrupt the probe states. In the following, we show that the errors caused by this unintended noise can be mitigated if it is of Pauli form. Specifically, we show that the framework of measurement error mitigation \cite{MZO:20:Quant} can be utilized to mitigate the errors cause by generalized Pauli channels. 

Let us first assume that $\mathcal{N}_{\mathrm{un}}$ is also of Pauli form. Then, we have the following convenient result. 
\begin{align}
		\mathcal{N}_{\mathrm{un}}\left( \rho \right) = \sum_{r = 0}^{d - 1} \sum_{s = 0}^{d-1} q_{r, s}W_{r, s}\rho W_{r, s}^{\dagger}.
\end{align}
Then $\mathcal{N}_{\mathrm{un}}$ commutes with $	\mathcal{N}_{\mathrm{dwc}}$. 

\begin{lemma}
	Let $\mathcal{N}_1$ and $\mathcal{N}_2$ be Pauli channels of the form \eqref{eq:DWCdef} with parameter set $\left\{p_{n,m}\right\}$ and $\left\{q_{r, s}\right\}$, respectively. Then for any quantum state $\rho$, $\mathcal{N}_1\circ \mathcal{N}_2 \left( \rho \right)= \mathcal{N}_2 \circ\mathcal{N}_1\left( \rho \right)$, where $\circ$ denotes the serial concatenation of two quantum channels.
	\label{lemma:commutation}
\end{lemma}
\begin{proof}
We can show this as follows
\begin{align}
		\mathcal{N}_{\mathrm{dwc}}\left(\mathcal{N}_{\mathrm{un}}\left( \rho \right)\right) 
		&= \sum_{n = 0}^{d - 1}\sum_{m = 0}^{d - 1} p_{n, m}
		W_{n,m}
			\left(\sum_{r = 0}^{d - 1} \sum_{s = 0}^{d-1} q_{r, s}W_{r, s}\rho W_{r, s}^{\dagger}\right)
		W_{n,m }^{\dagger} \nonumber \\
		&=
		\sum_{n = 0}^{d - 1}\sum_{m = 0}^{d - 1} \sum_{r = 0}^{d - 1} \sum_{s = 0}^{d-1}
		p_{n, m}q_{r, s}
		W_{n,m}
		 W_{r, s}\rho W_{r, s}^{\dagger}
		W_{n,m }^{\dagger}
		\label{eq:bef_comm}\\
		&=
		\sum_{r = 0}^{d - 1} \sum_{s = 0}^{d-1}\sum_{n = 0}^{d - 1}\sum_{m = 0}^{d - 1} 
		q_{r, s}p_{n, m}
		W_{r, s}
		W_{n,m}\rho W_{n,m }^{\dagger}
		W_{r, s}^{\dagger}
		\label{eq:aft_comm}	\\
		&= \sum_{r = 0}^{d - 1} \sum_{s = 0}^{d-1} q_{r, s}
		W_{r,s}
		\left(
		\sum_{n = 0}^{d - 1}\sum_{m = 0}^{d - 1} p_{n, m}W_{n, m}\rho W_{n, m}^{\dagger}
		\right)
		W_{r,s }^{\dagger}\nonumber \\
		&= \mathcal{N}_{\mathrm{un}}\left(\mathcal{N}_{\mathrm{dwc}}\left( \rho \right)\right). \nonumber 
\end{align}
Moving from \eqref{eq:bef_comm} to \eqref{eq:aft_comm}, we changed the order of summation, the order of product of $q_{r, s}$ and $p_{n, m}$, and also the order of product of Weyl operators by utilizing the commutation relation $\we{n,m}\we{r, s} = \omega^{rm - sn}\we{r, s}\we{n, m}$.
\end{proof}

The commutation of these two noisy channels allows use to model noise anywhere in the protocol by a single noisy process $\mathcal{N}_{\mathrm{un}}$ as long as its overall form is of a Pauli channel. Furthermore, for ease in the analysis we can move $\mathcal{N}_{\mathrm{un}}$ to any point in the protocol before measurement. However, it makes more sense to assume that $\mathcal{N}_{\mathrm{un}}$ acts only on the probe states before leaving the Alice's laboratory. That is because all noisy evolution after leaving Alice's laboratory and before being measured by Bob is actually the noisy channel between Alice and Bob. 

Then, Alice can execute the DPEPC locally in her laboratory to estimate the parameters of $\mathcal{N}_{\mathrm{un}}$ and send this information classically to Bob, who can utilize the measurement error mitigation framework as described below. 

Errors caused by a faulty measurement device are termed as measurement errors and are characterized by a column stochastic matrix $\Gamma$ \cite{MZO:20:Quant}. Let us assume that we apply a projective measurement characterized by a set of projectors $\left\{\Pi_i\right\}_i$ on a quantum state $\rho$. The ideal probabilities of measurement outcomes are given by a probability vector $P^{\mathrm{ideal}}$ whose $i$th element is $p_{i} = \tr\left(\Pi_i \rho\right)$. On the other hand, the probabilities of measurement outcome from a noisy measurement device characterized by $\Gamma$ are given by a probability vector given by $P^{\mathrm{noisy}} = \Gamma P^{\mathrm{ideal}}$. If the noise in the measurement device is known, i.e., if $\Gamma$ is known, an estimate of ideal probabilities of measurement outcomes can be obtained from noisy measurement results by $P^{\mathrm{ideal}} = \Gamma^{-1} P^{\mathrm{noisy}}$ \cite{MZO:20:Quant}.

By the virtue of Lemma~\ref{lemma:commutation}, we can assume $\mathcal{N}_{\mathrm{un}}$ to act just before the measurement. Let $\rho$ be the state before $\mathcal{N}_{\mathrm{un}}$ and the final measurement be in $\mcb{n, m}$. We can decompose $\rho$ in $\mcb{n,m}$ as
\begin{align}
	\rho = \sum_{i, j} \alpha_{i, j}\ket{i_{n, m}}\bra{j_{n, m}} = \sum_{i}\alpha_{i,i}\ket{i_{n, m}}\bra{i_{n, m}} + \sum_{\substack{i, j\\ i \neq j} } \alpha_{i, j}\ket{i_{n, m}}\bra{j_{n, m}},
\end{align}
where $\alpha_{i, i}$ is the probability of obtaining the measurement outcome corresponding to $\ket{i_{n, m}}$ when measuring $\rho$. We can write the state after $\mathcal{N}_{\mathrm{un}}$ as
\begin{align}
	\mathcal{N}\left( \rho\right) = \sum_{i, j} \alpha_{i, j}\mathcal{N}\left( \ket{i_{n, m}}\bra{j_{n, m}}\right) = \sum_{i}\alpha_{i,i}\mathcal{N}\left( \ket{i_{n, m}}\bra{i_{n, m}}\right) + \sum_{\substack{i, j\\ i \neq j} } \alpha_{i, j}\mathcal{N}\left( \ket{i_{n, m}}\bra{j_{n, m}}\right),
\end{align}
where we have dropped the subscript of $\mathcal{N}_{\mathrm{un}}$ for simplicity. 

The nondiagonal part of $\rho$ in the basis $\mcb{n, m}$, i.e., $ \sum_{\substack{i, j\\ i \neq j} } \alpha_{i, j}\ket{i_{n, m}}\bra{j_{n, m}}$ remains nondiagonal after the application of $\mathcal{N}_{\mathrm{un}}$ and does not contribute in the final measurement in $\mcb{n, m}$ \cite{RJK:18:SR}. Therefore, we only need to consider the $\mathcal{N}\left( \ket{i_{n, m}}\bra{i_{n, m}}\right)$ terms. Furthermore, the effect of DWC on the eigenstates of any $\we{n, m}$ followed by a measurement in $\mcb{n, m}$ can be modeled by a classical symmetric channel  \cite{RJK:18:SR}, which in turn is characterized by a doubly stochastic matrix. Therefore, 
\begin{align}
	\sum_{i}\alpha_{i,i}\mathcal{N}\left( \ket{i_{n, m}}\bra{i_{n, m}}\right)  = \sum_{i}\beta_{i,i} \ket{i_{n, m}}\bra{i_{n, m}}, 
\end{align}
where $\vec{\beta} = \Lambda_{n, m} \vec{\alpha}$, with $\vec{\alpha}, \vec{\beta},$ and $\Lambda_{n, m}$ as the vector of $\alpha_{i,i}$'s, vector of $\beta_{i, i}$'s, and the doubly stochastic matrix characterizing the effect of classical symmetric channel induced on $\mcb{n, m}$.

In case of a \emph{noiseless measurement device}, but the presence of  channel noise $\mathcal{N}_{\mathrm{un}}$ of Pauli form, we record the measurement  probabilities $\vec{\beta}$. However, if $\mathcal{N}_{\mathrm{un}}$ is known, we can simply estimate the ideal probabilities $\vec{\alpha} = \Lambda_{n, m}^{-1} \vec{\beta}$.

In case of a noisy measurement device as well as the presence of channel noise $\mathcal{N}_{\mathrm{un}}$ of Pauli form, we record the measurement probabilities $\vec{\gamma} = \Gamma \vec{\beta} = \Gamma \Lambda_{n,m} \vec{\alpha}$. Since the product of a left stochastic and a doubly stochastic matrix is another left stochastic matrix, we are still operating in the framework of measurement errors, and can perform the error mitigation as easily by inverting the matrix. Therefore, we can mitigate the errors caused by the noisy measurement device as well as from Pauli noise in the system in a unified manner.

Before moving to the numerical examples section, we remark that the only assumptions we made are the channel noise to be of Pauli form and the final measurement to be in the eigenbasis of some Pauli operator. These assumptions are not too demanding given the general nature of Pauli channels and the importance of Pauli measurements. Examples include the current protocol, quantum state tomography tasks \cite{TNW:02:PRA, GLF:10:PRL}, variational quantum algorithms \cite{PMS:14:NC}, and other quantum information processing tasks \cite{FL:11:PRL}, where Pauli measurements have the central role. Therefore, this modeling of Pauli noise in the framework of measurement errors and measurement error mitigation can be of independent interest beyond the protocol at hand. 

\section{Numerical Examples}\label{sec:NE}

In this Section we provide numerical examples of DPEPC and compare its performance with the entanglement-based optimal parameter estimation (OPE) method of \cite{FI:03:JPAMG} shown in Figure~\ref{fig:OPE}. The channel parameters were the eigenvalues of the $d^2 \times d^2$ exponential correlation matrix \cite{SW:08:T_IT}
\begin{equation}
\Phi \left( \gamma\right) 
= \frac{1}{d^2}
\left[ \gamma^{\sNorm{i - j}}\right]_{0 \leq i,j \leq d^2 - 1}.
\label{eq:corrMatrix}
\end{equation}
We recall that $\gamma = 0$ in \eqref{eq:corrMatrix} gives completely depolarizing (highly noisy) channel, and $\gamma = 1$ gives an ideal (noiseless) DWC. Furthermore, increasing $\gamma$ makes the channel parameters more ordered in terms of majorization, giving less noisy channels \cite{RJS:19:PRA, SW:08:T_IT}. Also, we assume the unintended channel to be the depolarizing channel parameterized by a real parameter $\kappa$, where $\kappa = 0, 1$ corresponds to the noiseless and the fully depolarizing channels, respectively.

\begin{figure*}[t]
	\centering
	\subfigure[$d = 5$]{
		\includegraphics[width=0.47\textwidth]{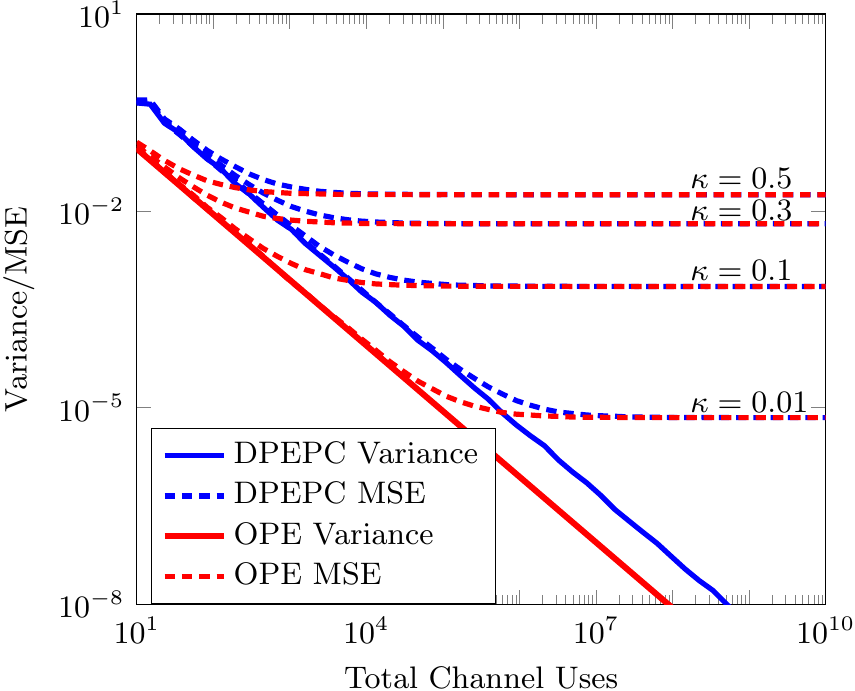}
	}
	\subfigure[$d = 6$]{
		\includegraphics[width=0.47\textwidth]{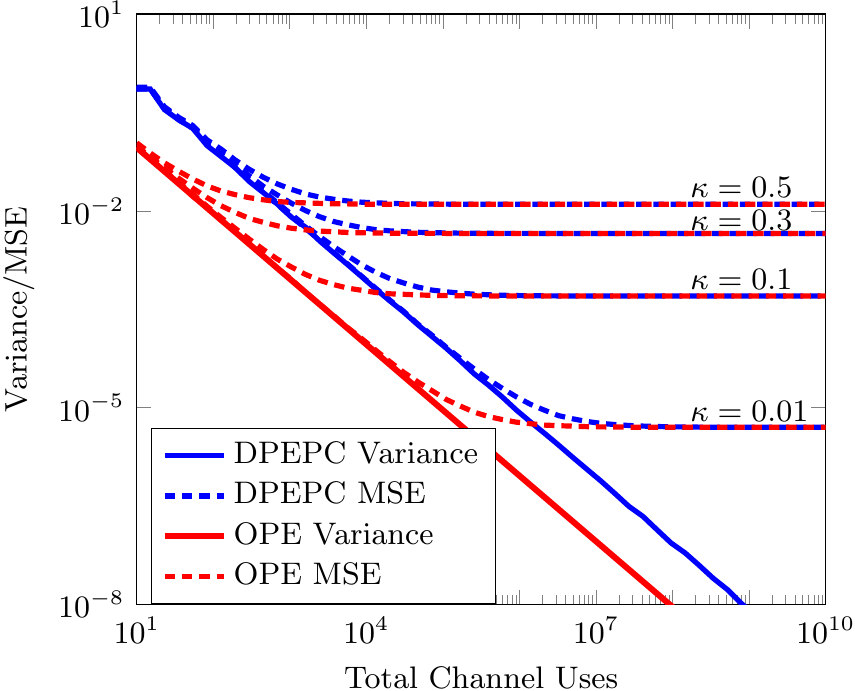}
	}
	\subfigure[$d = 7$]{
		\includegraphics[width=0.47\textwidth]{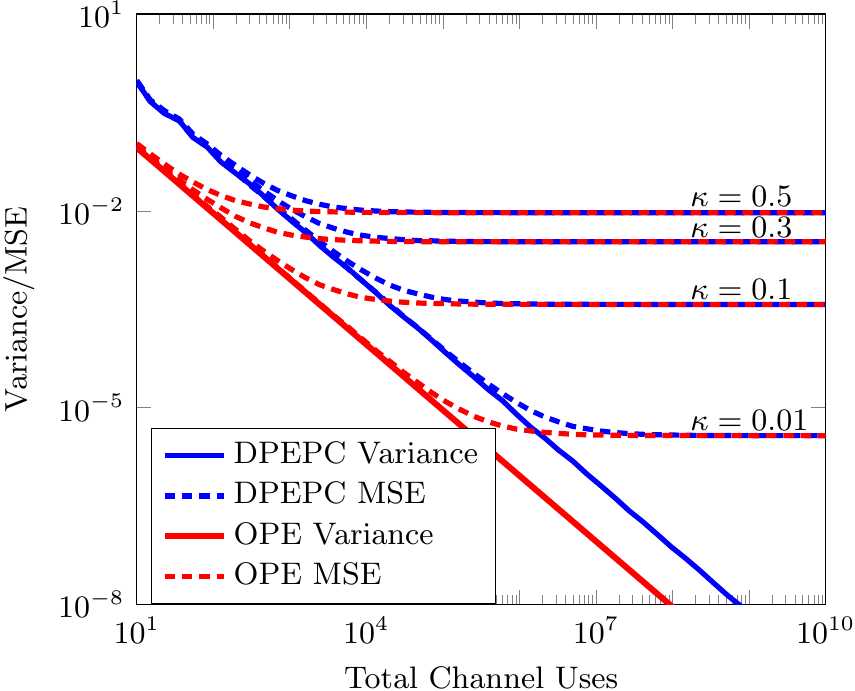}
	}
	\subfigure[$d = 8$]{
		\includegraphics[width=0.47\textwidth]{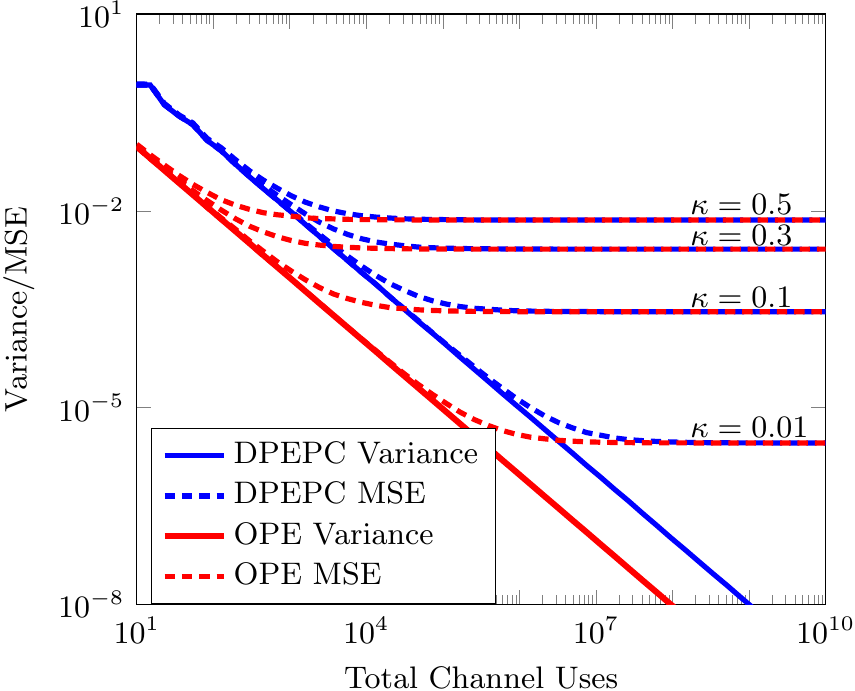}
	}
	\caption{\textbf{Performance comparison of DPEPC (proposed) and that of OPE (\cite{FI:03:JPAMG}).} We plot the number of channel uses vs the estimation accuracy for the DPEPC and the OPE of DWC with $\gamma = 0.7$, (a) $d = 5, K = 6$, (b) $d = 6, K = 12$, (c) $d = 7, K = 8$, and (d) $d = 8, K = 12$, with different values of the noise strength $\kappa$ in the probe states. {Both the variance and the MSE are summed over all $\hat{p}_{n,m}$.} We do not plot the MSE for $\kappa = 0$ for both the DPEPC and the OPE since they are exactly same as their respective variances. }
	\label{fig:noisyDPT}
\end{figure*}

{The performance metrics we use in our numerical examples are the variance and the mean square error (MSE) of the estimates. A natural performance metric for process tomography, channel estimation, and channel distinguishing problems is the diamond norm distance. However, we provide the results in the main text in terms of variance and the MSE because of the following reasons. 1) We are dealing with a parametric class of channels where the channel structure is fixed. Then, the problem essentially boils down to a parameter estimation problem. In these problems variance and the MSE are more natural performance metrics. 2) Together, MSE and the variance of the estimates provide more information. For example, We can easily observe a bias in our estimates if the MSE and the variance are not equal. For interested readers, we also provide all our numerical results in terms of diamond norm distance in the Appendix of this paper. }

Figure~\ref{fig:noisyDPT} shows the performance comparison of DPEPC and OPE of DWCs for $d = 5, 6, 7,$ and 8, and $\gamma = 0.7$ with different noise strengths $\kappa$. We plot the variance/MSE against the number of channel uses $N$. Blue (resp. red) Solid lines show the variance of DPEPC (resp. OPE), which is same for the MSE for the noiseless ($\kappa = 0$) case. {These values of variance and MSE are summed over all parameters $p_{n,m}$ of the channels.} Note that the two variance lines are parallel, depicting same scaling in variance $1/N$ as a function of number of channel uses ($N$). The separation between the two lines is the multiplicative factor in the scaling and is a function of $K$, the number of measurement configurations we need to uniquely identify all channel parameters. {This separation also shows the tradeoff between entanglement-assisted and entanglement-free schemes. By avoiding the use of entanglement in our scheme for the sake of experimental feasibility, we need to utilize more experimental configurations and perform our experiment more number of times (by a constant factor, independent of $N$) to obtain the same performance. }

{The performance of OPE for different values of $d$ looks very similar in Figure~\ref{fig:noisyDPT}. This seems counter intuitive but can be easily explained as follows. The measurement outcomes in the OPE follow a multinomial distribution where the probability of obtaining measurement outcome $i$ is $p_i$.\footnote{We use a single index here for the ease of notation.} Let $X_i$ be the random variable characterizing the number of times event $i$ is observed in $N$ trials. Then, the variance $\Var{X_i} = N p_i \left( 1 - p_i\right)$. Since we use the maximum likelihood estimator, i.e., $\hat{p}_i = X_i/N$, its variance is $\Var{\hat{p}_i} = \Var{X_i/N} = p_i \left( 1 - p_i\right)/N$. Summed variance of these estimators is $V = \sum_i p_i \left( 1 - p_i\right)/N$, which is \emph{independent} of $d$. Its dependence on the distribution is through $V_{\mathrm{e}} = \sum_i p_i \left( 1 - p_i\right)$, which changes very little for similarly generated distributions. For example, for the considered examples $V_\mathrm{e} = 0.960, 0.972, 0.980, 0.984$  for $d = 5, 6, 7,$ and 8, respectively. Uniform sampling from the probability simplex also gives similar values, i.e.,  $V_{\mathrm{e}} = 0.923, 0.946, 0.960, 0.970$ for $d = 5, 6, 7,$ and 8, respectively. Due to these small changes in $V_{\mathrm{e}}$ for increasing $d$ the variance of OPE looks very similar in all graphs.}

{The total variance of DPEPC can be calculated as $\tr\left( \M{\Sigma}_\V{x}\right)$, where the covariance matrix $\M{\Sigma}_\V{x} = \M{B}_{K}^d \M{\Sigma}_{\V{b}_K^d} \left( \M{B}_{K}^d\right)^T$, where $\M{B}_{K}^d$ is from \eqref{eq:B_matrix}, and $\M{\Sigma}_{\V{b}_K^d}$ is the block diagonal covariance matrix of $\V{b}_K^d$. We can see the aforementioned effect in the variance of DPEPC when $K$ is same, e.g., for $d = 6$ and 8. On the other hand, the effect of increase ($d = 5, K = 6$ to $d = 6, K = 12$) or decrease ($d = 6, K = 12$ to $d = 7, K = 8$) in $K$ affects the performance as expected. }


Dashed lines in Figure~\ref{fig:noisyDPT} show the performance of DPEPC and OPE when the initial probe states are subject to depolarizing noise of strength $\kappa$. It can be seen that stronger the noise, earlier the MSE departs from the variance of the estimators and becomes independent of $N$.

\begin{figure}[t!]
	\centering
	\includegraphics[width=0.65\textwidth]{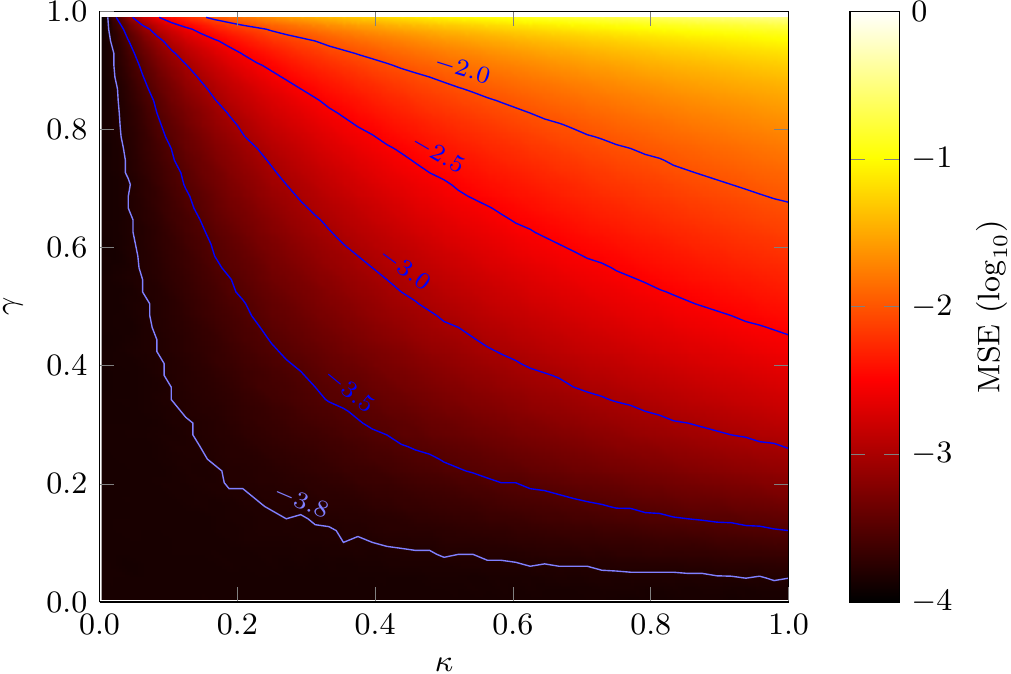}
	\caption{\textbf{Noise tolerance of DPEPC (proposed).} We plot the MSE of DPEPC as a function of noisiness $\gamma$ of the channel under study and the strength $\kappa$ of the depolarizing noise on the probe states for $d = 13, K = 14$, and  $N = 10^5$. The MSE is least affected by the unintended noise of strength $\kappa$ for small values of $\gamma$. }
	\label{fig:gammaEffect}
\end{figure}

It is natural to think that noisier the original DWC whose description we are trying to obtain, higher the tolerance to the depolarizing noise. For example, if the original DWC is completely depolarizing, the MSE will improve with increasing $N$ indefinitely, regardless of the depolarizing noise strength $\kappa$. {Figure~\ref{fig:gammaEffect} confirms this intuition, where we plot the MSE of DPEPC as a function of $\gamma$ and $\kappa$. We recall that $\gamma = 0$ gives a completely depolarizing channel and $\gamma = 1$ is a noiseless channel. From the figure, it is clear that the noise on probes has a minimal effect on the MSE performance for smaller $\lesssim 0.2$ values of $\gamma$. On the other hand, closer the channel under study is to the ideal one, it is more affected by the unintended noise of strength $\kappa$.}

The effect of depolarizing noise on the initial probe states can be interpreted as a noise strength-dependent saturation point on the number of channel uses $N$, such that increasing $N$ does not improve the estimation MSE beyond the saturation point. This saturation behaviour is typical of a biased estimator, i.e., 
\begin{align}
	\EX{\hat{p}_{n,m}} = p_{n,m} + \nu \neq p_{n,m},
	\label{eq:bias}
\end{align}
where \( \nu \neq 0\) is the bias in the estimates. If the strength of noise on the initial probe states is known, it is possible to avoid the saturation behaviour seen in Figure~\ref{fig:noisyDPT} and \ref{fig:gammaEffect} by utilizing the measurement error mitigation framework.   For the depolarizing channel, we can achieve this by simply setting
\begin{align}
\tilde{p}_{n,m} = \frac{1}{1 - \kappa} \left( \hat{p}_{n,m} - \frac{\kappa}{d^2}\right),
\label{eq:correction}
\end{align}
where $\tilde{p}_{n,m}$ is the new (bias mitigated) estimate of $p_{n,m}$. Note that since the depolarizing channel is a special case of the Pauli channel,  \eqref{eq:correction} is a special case of the measurement error mitigation by matrix inversion, which was simplified due to the high symmetry of the depolarizing channel. 

\begin{figure}[t]
	\centering
	\includegraphics[width=0.55\textwidth]{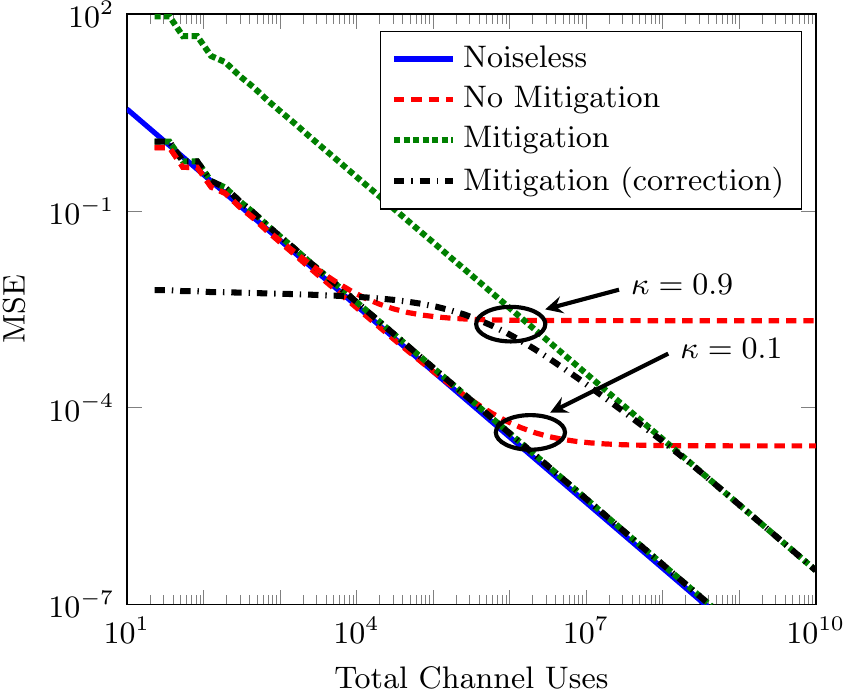}
	\caption{\textbf{Quantum error mitigation for recovering the original scaling.} Effect of noise and the result of error mitigation in DPEPC (proposed) of a DWC with {$d = 27, K = 36$},\, $\gamma = 0.7$, and $\kappa = 0.1$ and $0.9$. Error mitigation successfully recovers the original scaling of noiseless DPEPC, with a constant noise-dependent multiplicative factor. }
	\label{fig:biasMitigation}
\end{figure}

Figure~\ref{fig:biasMitigation} shows the effect of error mitigation in the DPEPC estimates of a DWC in {$d=27$} with $\gamma = 0.7$, and $\kappa = 0.1$ and 0.9. It can be noted that the error mitigation introduces a $\kappa$-dependent multiplicative factor in the scaling against the number of channel uses $N$. This is because for large $\kappa$, contribution from uniform distribution dominates in \eqref{eq:bias}, reducing the information about the original distribution in the measurement outcomes. Since we are utilizing maximum likelihood estimates, it is possible to obtain incorrect channel parameters, i.e., negative elements, and parameter sum not equal to one. The effect of these incorrect parameter ranges is enhanced due to error mitigation. In such cases, we set the negative parameter values to 0 and normalize the error mitigated distribution, i.e., we project the obtained vector on the probability simplex. We call this process correction and plot the MSE performance of error mitigation with both correction and without correction. It can be seen that this correction significantly improves the MSE performance of bias mitigated estimates.

Finally, the main ingredient of our DPEPC for DWCs is Lemma~\ref{Lemma:1}, which has been generalized to the generalized Pauli channels in \cite{SIU:20:JPAMT}. From discussion in \cite{SIU:20:JPAMT}, it is straightforward to generalize the DPEPC for other generalizations of Pauli channels. 

\section{Conclusions} \label{sec:Conc}
We have presented a process tomography/parameter estimation scheme for DWCs, which can be extended to the other definitions of generalized Pauli channels. The proposed method operates with separable probe states, yet provides same MSE scaling against number of channel uses as that of entangled-based parameter estimation scheme. Numerical examples show that the number of measurement configurations $K$ scales linearly with the Hilbert space's dimension. 
We also {showed that the framework of measurement error mitigation can be useful in systems with Pauli noise and Pauli measurements.} In particular, we exemplified the depolarizing noise on the probe state of both DPEPC and OPE and mitigated the consequent errors by utilizing measurement error mitigation. 
Future directions may include analytical results on $K$, the number of measurement configurations required in $\HS_d$. 
Furthermore, the DPEPC and the OPE are clearly the extreme points in terms of utilizing entanglement in parameterized channel estimation task. It needs to be investigated if there exist some intermediate schemes where limited entanglement may be utilized to access the tradeoff between the entanglement and the number of channel uses/required number of experimental configurations.
Another possible future direction is to utilize DPEPC as a first step in identifying \emph{any} unknown given channel and then using the results of DPEPC in a second step to completely identify the unknown channel. This direction can particularly be interesting due to low resource requirements and fast converging behaviour of DPEPC. 


\section*{Acknowledgments}
This work was supported by the National Research Foundation of Korea (NRF) grant funded by the Korea government (MSIT) (No. 2019R1A2C2007037).
\bibliographystyle{unsrtnat}

\pagebreak

\appendix

\setcounter{equation}{0}
\setcounter{figure}{0}
\setcounter{table}{0}
\setcounter{page}{1}
\makeatletter
\newtagform{supplementary}[S]()
\usetagform{supplementary}

\renewcommand{\thefigure}{S\arabic{figure}}
\renewcommand{\bibnumfmt}[1]{[S#1]}
\renewcommand{\thefigure}{S\arabic{figure}}

\section{Diamond Norm Distance Performance}\label{app:DND}
The diamond norm distance is a natural distance between two quantum channels. For the sake of completeness and for  interested readers, here we reproduce all numerical examples figures of main text with diamond norm distance as the performance metric. 

The diamond norm distance between two quantum channels $\mathcal{N}_1$ and $\mathcal{N}_2$ is given by \cite{Kit:97:RMS, WE:16:PRA}
\begin{align}
	\vNorm{\mathcal{N}_1 - \mathcal{N}_2}_{\diamond} 
	= \sup_{\psi} \vNorm{
	\left( \mathcal{N}_1 \otimes \mathcal{I}_d - 	\mathcal{N}_2 \otimes \mathcal{I}_d \right) 
	\left( \psi \right)
	}_1,
\label{eq:DN}
\end{align}
where $\mathcal{I}_d$ is the identity map on $\HS_d$, and $\vNorm{X}_1 = \sqrt{\tr\left( X^{\dagger} X\right)} $. 

The diamond norm distance naturally captures the notion of distinguishability of two quantum channels \cite{BS:10:JPB}. The diamond norm distance can be formulated as a semidefinite convex program and thus can be efficiently computed in the problem size \cite{Wat:09:ToC, BT:10:QIC}. Despite its efficiency of computation as a convex program, it is difficult to employ semidefinite programming in the present manuscript since we have examples where $d$ is as large as 27. Consequently, the problem size of optimization in \eqref{eq:DN} is $d^2 = 729$ which is not easy to solve on a personal computer. Furthermore, for obtaining good quality numerical examples, averaging of several samples on each point is required. Due to these reasons, utilizing convex programming for providing numerical examples of this paper is difficult. 

Fortunately, for the case of Pauli channels, an exact analytical expression for diamond norm distance can be obtained \cite{Sac:05:PRA, BS:10:JPB}. Let $\mathcal{N}_1$ and $\mathcal{N}_2$ be two Pauli channels of arbitrary finite dimension with parameter sets $\left\{ p_{n, m}\right\}$ and $q_{n, m}$, then \cite{Sac:05:PRA, BS:10:JPB}
\begin{align}
	\vNorm{\mathcal{N}_1 - \mathcal{N}_2}_{\diamond} = \sum_n\sum_m \sNorm{p_{n, m} - q_{n, m}}.
	\label{eq:DN_Simple}
\end{align}

We use \eqref{eq:DN_Simple} to calculate the diamond norm distance between the actual and estimated Pauli channels. We replicate all the numerical example figures from the main text, i.e., Figs.~\ref{fig:noisyDPT}, \ref{fig:gammaEffect}, and \ref{fig:biasMitigation} as Figs.~\ref{fig:noisyDPTS}, \ref{fig:gammaEffectS}, and \ref{fig:biasMitigationS}, respectively. These supplementary figures provide the same qualitative insights as that of main text numerical examples except Figure~\ref{fig:noisyDPTS}. 

In Figure~\ref{fig:noisyDPTS}, we note that increase in $d$ has slightly more perceptible difference as compared to what we noticed in Figure~\ref{fig:noisyDPT}. This is because \eqref{eq:DN_Simple} is actually the $\ell_1$ norm distance between the actual and the estimated distribution and is of order $\approx \sum_i\sqrt{p_i \left( 1 - p_i \right)/N}$ for maximum likelihood estimates of the distribution $\left\{ p_i\right\}$ \cite{HJW:15:T_IT}. The term $\sum_i\sqrt{p_i \left( 1 - p_i \right)} = 4.07, 4.93, 5.79,$ and $6.63$ for $d = 5, 6, 7,$ and 8, respectively, for the considered examples of $\gamma = 0.7$. For uniform sampling from the probability simplex, this term is $\approx 4.32, 5.22, 6.12,$ and  7.02 for $d = 5,6,7,$ and 8, respectively. Due to this slightly higher increase in these values for increasing $d$, the dependence on $d$ in Figure~\ref{fig:noisyDPTS} is slightly more perceptible than in the variance in Figure~\ref{fig:noisyDPT} of the main text.

\setcounter{figure}{4}

\begin{figure*}[t]
	\centering
	\subfigure[$d = 5$]{
		\includegraphics[width=0.47\textwidth]{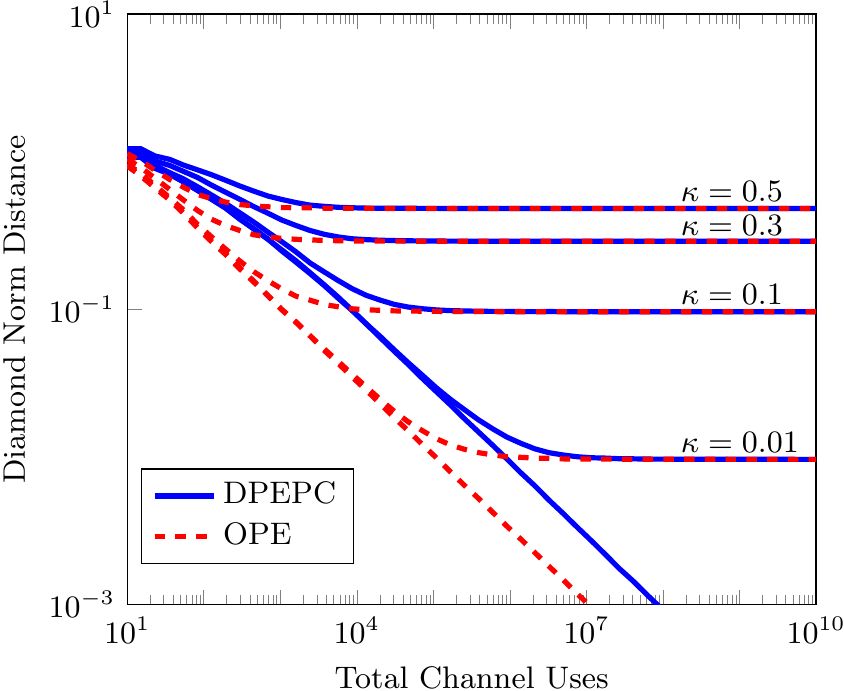}
	}
	\subfigure[$d = 6$]{
		\includegraphics[width=0.47\textwidth]{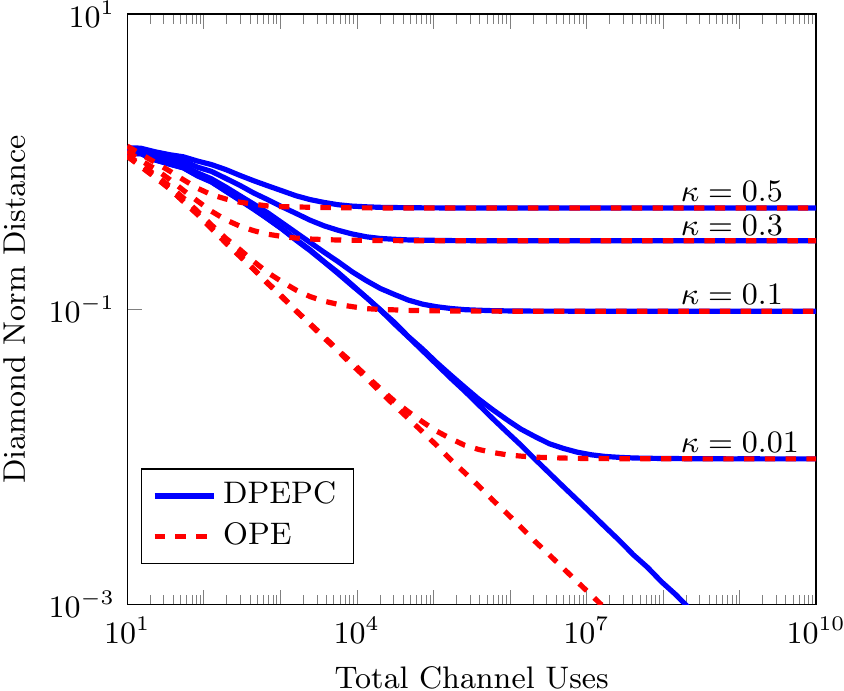}
	}
	\subfigure[$d = 7$]{
		\includegraphics[width=0.47\textwidth]{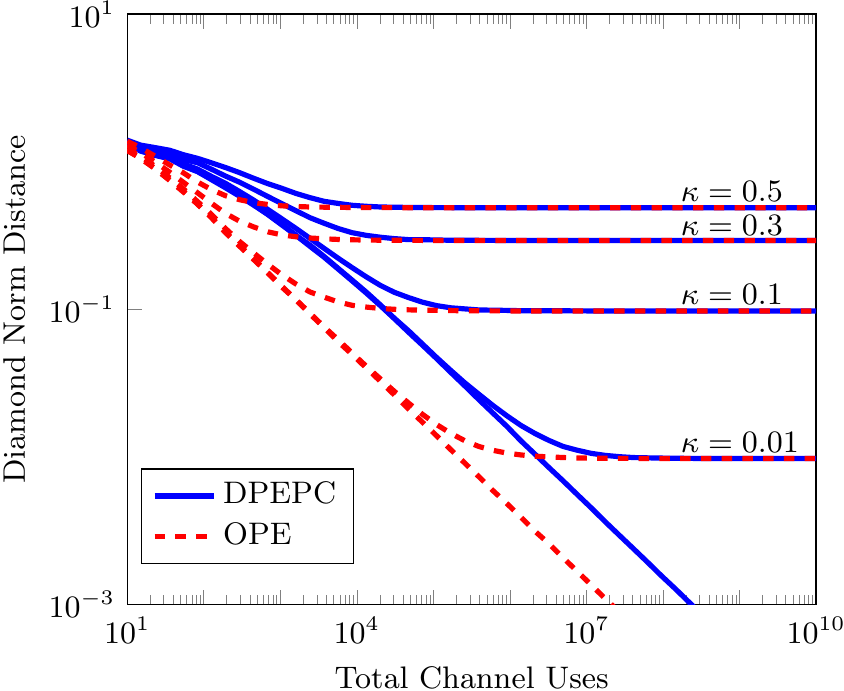}
	}
	\subfigure[$d = 8$]{
		\includegraphics[width=0.47\textwidth]{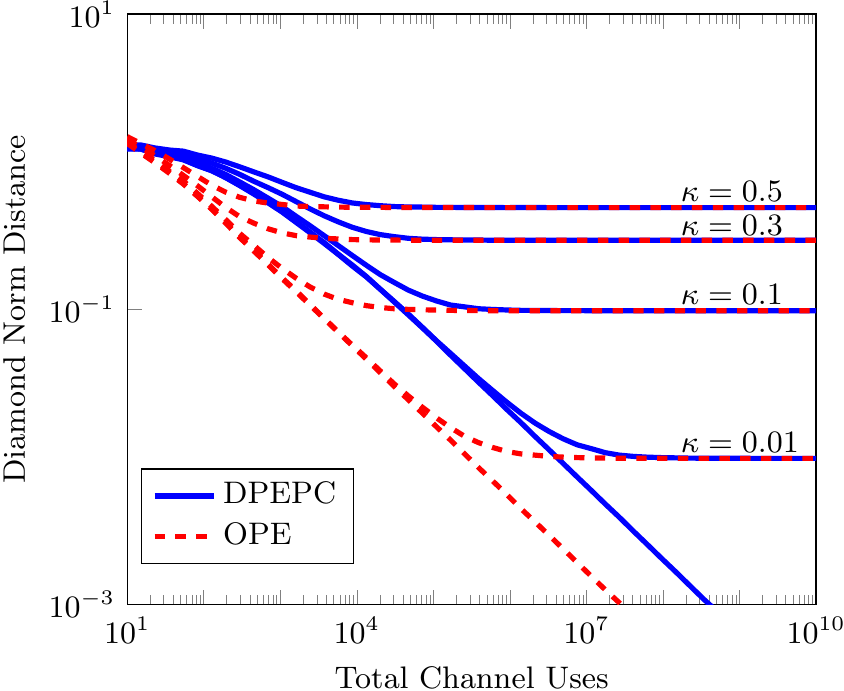}
	}
	\caption{\textbf{Performance comparison of DPEPC (proposed) and that of OPE (\cite{FI:03:JPAMG}).} We plot the number of channel uses vs the estimation accuracy (diamond norm distance) for the DPEPC and the OPE of DWC with $\gamma = 0.7$, (a) $d = 5, K = 6$, (b) $d = 6, K = 12$, (c) $d = 7, K = 8$, and (d) $d = 8, K = 12$, with different values of the noise strength $\kappa$ in the probe states.}
	\label{fig:noisyDPTS}
\end{figure*}

\begin{figure}[ht!]
	\centering
	\includegraphics[width=0.65\textwidth]{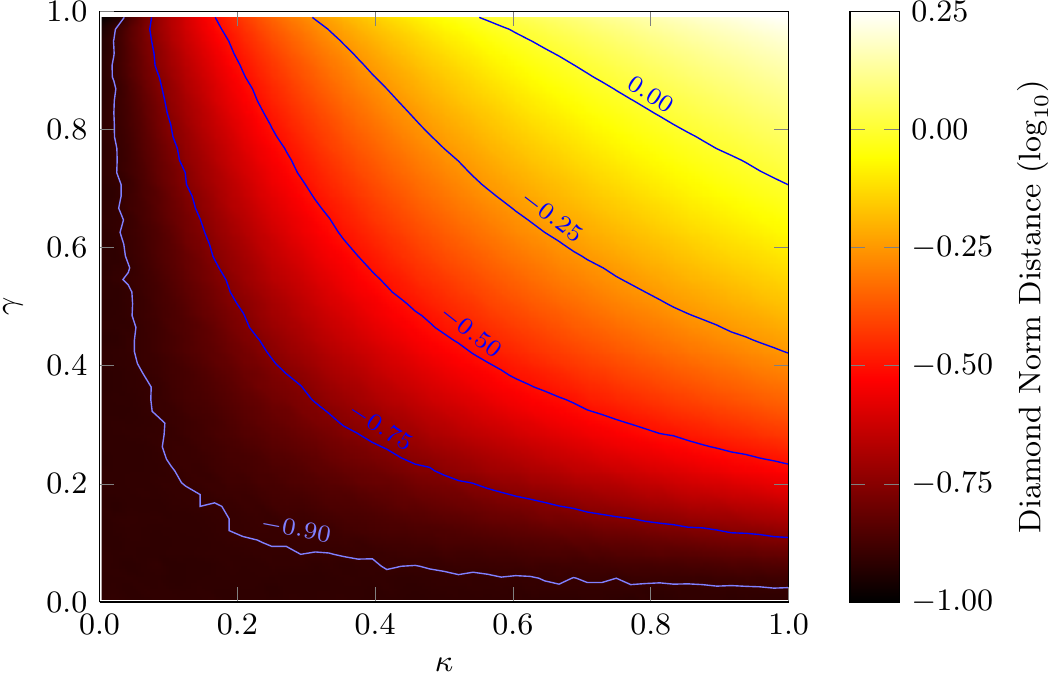}
	\caption{\textbf{Noise tolerance of DPEPC.} We plot the diamond norm distance of estimated via DPEPC and the actual channel as a function of noisiness $\gamma$ of the channel under study and the strength $\kappa$ of the depolarizing noise on the probe states for $d = 13, K = 14$ and $N = 10^5$. The MSE is least affected by the unintended noise of strength $\kappa$ for small calues of $\gamma$}
	\label{fig:gammaEffectS}
\end{figure}

\begin{figure}[ht!]
	\centering
	\includegraphics[width=0.55\textwidth]{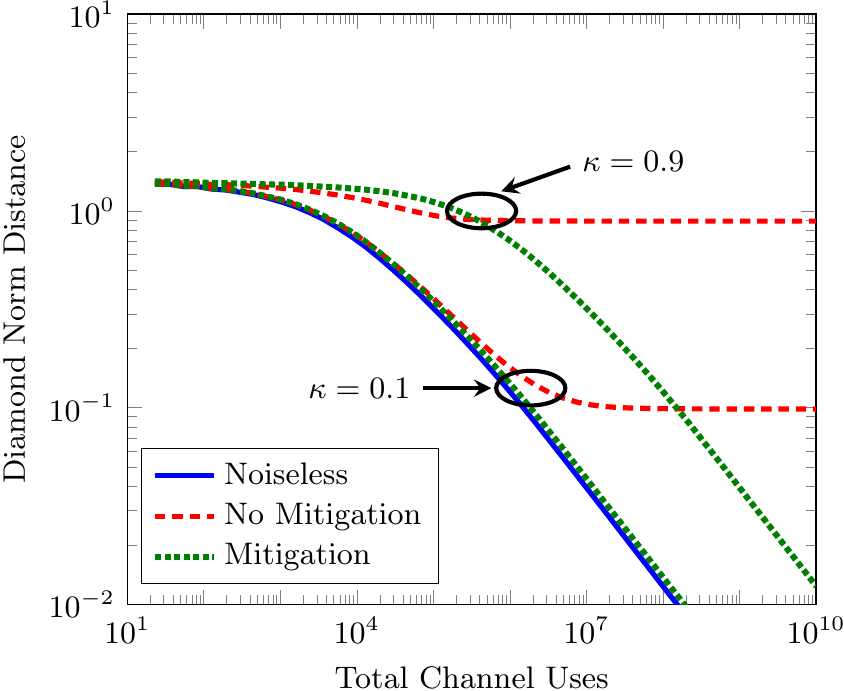}
	\caption{\textbf{Quantum error mitigation for recovering the original scaling.} Effect of noise and the result of error mitigation in DPEPC of a DWC with $d = 27$,\, $\gamma = 0.7$, and $\kappa = 0.1$ and $0.9$. Error mitigation successfully recovers the original scaling of noiseless DPEPC, with a constant noise-dependent multiplicative factor. }
	\label{fig:biasMitigationS}
\end{figure}

\end{document}